\documentclass{llncs}
\title{A Partially Ordered Structure and a Generalization of the Canonical Partition 
for General Graphs with Perfect Matchings}
\author{Nanao Kita}
\institute{Keio University, Yokohama, Japan\\
           \email{kita@a2.keio.jp}}
\usepackage{amsmath, amssymb}
\usepackage{algorithmic}
\usepackage{algorithm}
\usepackage{cite}
\usepackage{setspace}
\usepackage{latexsym}
\usepackage{color}
\usepackage{comment}

\newcommand{\matchablesp}{factorizable~}

\newcommand{\matchable}{factorizable}

\newcommand{\zero}{balanced~}

\newcommand{\exposed}{exposed}

\newcommand{\Vg}{V(G)}
\newcommand{\Eg}{E(G)}

\newcommand{\yield}{\triangleleft}
\newcommand{\gpart}[1]{\mathcal{P}(#1)}
\newcommand{\pargpart}[2]{\mathcal{P}_{#1}(#2)}
\newcommand{\gsim}{\sim_g}

\newcommand{\up}[1]{up(#1)}
\newcommand{\upstar}[1]{up^{*}(#1)}
\newcommand{\parup}[2]{up_{#1}(#2)}
\newcommand{\parupstar}[2]{up^{*}_{#1}(#2)}
\newcommand{\what}{\widehat}
\newcommand{\su}[2]{#1\oplus #2}
\renewcommand{\Gamma}{N}

\spnewtheorem*{proofof}{Proof of}{\itfamily}{\rmfamily}
\spnewtheorem{cclaim}{Claim}{\rmfamily}{\rmfamily}
\spnewtheorem*{ac}{Acknowlegements}{\bf}{\rmfamily}

\begin{document}
\maketitle

\begin{abstract}
This paper is concerned with  structures of 
general graphs with perfect matchings.
We first reveal a partially ordered structure 
among factor-components of 
general graphs with perfect matchings.
Our second result is  a generalization of Kotzig's canonical partition 
to a decomposition  
of general graphs with perfect matchings.
It contains a short proof for the theorem of the canonical partition. 
These results give decompositions   
which are canonical, that is, unique to given graphs.  
We also show 
that there are correlations between these two and that  
these can be computed in polynomial time. 
\end{abstract}

\section{Introduction}
%
This paper is concerned with matchings on graphs.
 For general accounts on matching theory we refer to  Lov\'asz and Plummer's book~\cite{lp1986}.

A {\em  matching} of a graph $G$ is a 
set of edges $F\subseteq E(G)$   
no two of which  have  common vertices.
A matching of cardinality $|V(G)|/2$ (resp. $|V(G)|/2 - 1$) is called 
a {\em  perfect matching} (resp. a {\em  near-perfect matching}). 
We call a graph with a perfect matching \textit{\matchable}.
An edge of a factorizable graph is called \textit{allowed}
if it is contained in a perfect matching. 
For a factorizable graph $G$, 
each connected component of the subgraph of $G$
determined by all the allowed edges of it
is called an \textit{elementary component} of $G$. 
A factorizable graph which has exactly one elementary component 
is called \textit{elementary}. 
For each elementary component $H$, 
we call $G[V(H)]$ a {\em factor-connected component} or 
{\em factor-component} of $G$, 
and denote the set of all the factor-components of $G$ as $\mathcal{G}(G)$. 
%

Matching theory is of central importance in 
graph theory and combinatorial optimization%
~\cite{schrijver2003}, 
with numerous  practical applications~\cite{cc2005}.
Structure theorems that give decompositions which are canonical, 
namely, unique to given graphs, 
play important roles in matching theory. 
Only three theorems, i.e.    
the canonical partition~\cite{kotzig1959a, kotzig1959b, kotzig1960}, 
the Dulmage-Mendelsohn decomposition~\cite{lp1986}, 
and the Gallai-Edmonds structure theorem~\cite{lp1986} have been  known as such.
%
%
The first two 
  are not applicable for general graphs with perfect matchings,
and the last one  treats them as  irreducible  
and does not decompose them properly,
which means nothing has been known that 
tells non-trivial canonical structures of 
general graphs with perfect matchings.
Therefore, in this paper, 
we give new canonical structure theorems for them.

By the definitions,  
we can view factorizable graphs as being ``built'' up by  combining 
factor-components with additional edges. 
However it does not  mean that all combinations 
result in graphs with 
desired factor-components.
Thus the family of factor-components must have 
a certain non-trivial structure. 
For bipartite \matchablesp graphs, 
the \textit{Dulmage-Mendelsohn decomposition} (in short, the \textit{DM-decomposition})
 reveals the ordered structure of their factor-components.
  However, as for non-bipartite graphs, no counterpart has been known.

%
In this paper, as our first contribution, 
we reveal a partially ordered structure between factor-components 
of general graphs with perfect matchings.
It has some similar natures to the DM-decomposition, 
however they are distinct.

The second contribution is a generalization of the 
\textit{canonical partition}~\cite{kotzig1959a, kotzig1959b, kotzig1960}; see also~\cite{lp1986}, which is originally a 
decomposition of  elementary graphs. 
Kotzig~\cite{kotzig1959a, kotzig1959b, kotzig1960} 
first investigated  the canonical partition of elementary graphs 
as the quotient set of a certain equivalence relation,
 and later, Lov\'asz redefined it from the point of view of maximal barriers~\cite{lp1986}.
%
In this paper we generalize the canonical partition to a decomposition 
of general graphs with perfect matchings,
based on Kotzig's way. 
It contains a short proof for the theorem of the canonical partition.
%

Note that these two results of us  give canonical decompositions of graphs.
We also show that there are correlations between these two 
 and that these  
can be computed in polynomial time.

Any of the three existing canonical structure theorems 
plays significant roles in combinatorics including matching theory.  
The canonical partition plays a crucial role in matching theory, 
especially from the polyhedral point of view, that is, in the study of 
the matching polytope and  the matching lattice~\cite{elp1982, lovasz1987,clm2003}.
The Dulmage-Mendelsohn decomposition is known for its application 
to the efficient solution of linear equations determined by large sparse matrices~\cite{lp1986}.
Additionally, it is an origin of a series of studies on submodular functions,
that is, the field of  the principal partition~\cite{nakamura1988, fujishige2005}.
The Gallai-Edmonds structure theorem is essential to  the optimality of 
the maximum matching~\cite{edmonds1965, lp1986}.
Thus it also underlies  reasonable generalizations 
of maximum matching problem~\cite{ps2004, ss2004}.

By combining the results in this paper with the Gallai-Edmonds structure theorem,
we can easily obtain a refinement of the Gallai-Edmonds structure theorem,
which gives a consistent view of graphs, whether
they are factorizable or not, or,  elementary or not~\cite{nanaokita}.
Hence,  we are sure that our structure theorems 
should be powerful tools in matching theory.
In fact, the cathedral theorem~\cite{lp1986}
can be obtained from our results in a quite natural way~\cite{nanaokita}.
%
%
%
%
%
%

%
\section{Preliminaries}
In this section, we list some standard definitions and  
well-known properties. 
Basics on sets, graphs, digraphs, and algorithms 
mostly conform to~\cite{schrijver2003}.

%

Let $G$ be a graph and $X\subseteq \Vg$. 
The subgraph of $G$ induced by $X$ is denoted by $G[X]$.
$G-X$ means $G[\Vg\setminus X]$. 
Given $F\subseteq E(G)$, we define the \textit{contraction} of $G$ by $F$ 
as the graph obtained from contracting all the edges in $F$, 
and denote as $G/F$. 
Additionally, We define the \textit{contraction} of $G$ by $X$ as 
$G/X := G/E(G[X])$.
We say $H\subseteq G$ if $H$ is a subgraph of $G$. 
If it is clear from the context, we sometimes regard a subgraph $H\subseteq G$ 
as the vertex set $V(H)$,
 a vertex $v$ as a graph $(\{v\}, \emptyset )$.

The set of edges that has one end vertex  in $X\subseteq V(G)$ and 
the other vertex in $Y\subseteq V(G)$ is denoted as $E_G[X, Y]$.
We denote $E_G[X, V(G)\setminus X]$ as $\delta_G(X)$.
We define the \textit{set of neighbors}  of $X$ as
the set of vertices in $V(G)\setminus X$ that 
are adjacent to vertices in $X$, and denote as $\Gamma_G(X)$.
We sometimes denote $E_G[X, Y]$, $\delta_G(X)$, $\Gamma_G(X)$ as just  
$E[X,Y]$, $\delta(X)$, $\Gamma(X)$
if they are apparent from the context.

For two graphs $G_1$ and $G_2$,
$G_1+G_2 := \left( V\left( G_1 \right) \cup V\left( G_2\right),
 E\left( G_1\right)\cup E\left( G_2\right) \right)$ is called
 the \textit{union} of them,
and  $G_1\cap G_2 := \left( V\left( G_1\right)\cap V\left( G_2\right),
 E\left( G_1\right) \cap E\left( G_2\right) \right) $
  the \textit{intersection} of them.

Let $\hat{G}$ be a graph such that $G\subseteq \hat{G}$.
For $e = uv\in E(\hat{G})$,
$G+e$ means $(V(G)\cup \{u, v\}, E(G)\cup\{e\})$,
and $G-e$ means $(V(G), E(G)\setminus \{e\})$.
For a set of edges $F = \{e_i\}_{i=1}^k$,
$G +F$ and $G -F$ means respectively $G + e_1 + \cdots  + e_k$ and  $G - e_1 -\cdots - e_k$.

For a path $P$ and $x, y \in V(P)$,
$xPy$ means the subpath on $P$ between $x$ and $y$.
For a circuit $C$ with an orientation that makes it a dicircuit,
and $x, y\in V(C)$ where $x\neq y$,
$xCy$ means the subpath in $C$ that can be regarded 
as a dipath from $x$ to $y$. 
%

A vertex $v\in V(G)$ satisfying $\delta(v)\cap M = \emptyset $ is called 
\textit{exposed} by $M$.
For  a matching $M$ of $G$ and $u\in V(G)$,
$u'$ denote the vertex to which $u$ is matched by $M$.
For $X\subseteq V(G)$, 
$M_X$ denotes $M\cap E(G[X])$.

Let $M$ be a matching of $G$.
For $Q\subseteq G$, which is a path or circuit,
we call $Q$  $M$-{\em  alternating}
if $E(Q)\setminus M$ is a matching of $Q$.
Let $P\subseteq G$ be an $M$-alternating path with end vertices $u$ and $v$.
If $P$ has an even number of edges 
 and starts with an edge in $M$ if it is traced  from $u$, 
we call it an $M$-{\em  \zero path} from $u$ to $v$.
We regard a trivial path, that is, a path composed of 
one vertex and no edges  as an $M$-\zero path.
If $P$ has an odd number of edges and $M\cap E(P)$ (resp. $E(P)\setminus M$) is a perfect matching of $P$,
we call it $M$-{\em  saturated} (resp. $M$-{\em  \exposed}).
%
%

Let $H\subseteq G$.
 We say a path $P\subseteq G$  is an {\em  ear relative to} $H$ 
if both end vertices of $P$ are in $H$ while internal vertices are not.
So do we to a circuit  if exactly one vertex of it is in $H$.
For simplicity, we call the vertices of $V(P)\cap V(H)$ 
{\em end vertices} of $P$, even if $P$ is a circuit.
For an ear $R\subseteq G$ relative to $H$,
we call it an $M$-{\em  ear}
if  $P-V(H)$ is an $M$-saturated path. 

%
%
%
A graph is called {\em  factor-critical} if 
any deletion of its single vertex leaves a factorizable graph.
A subgraph $G'\subseteq G$ is called {\em  nice} if 
$G-V(G')$ is factorizable.
The next two propositions are  well-known and
might be regarded as folklores.
\begin{proposition}\label{prop:path2root}
Let $M$ be a near-perfect matching of a graph $G$ that exposes $v\in \Vg$. 
Then, $G$ is factor-critical if and only if for any $u\in \Vg$ there exists 
an $M$-\zero path from $u$ to $v$.
\end{proposition}
\begin{proposition}\label{prop:fc_block}
Let $G$ be a graph.
Then $G$ is factor-critical if and only if each block of $G$ is factor-critical.
\end{proposition}
\begin{proposition}[\textmd{implicitly stated in~\cite{lovasz1972a}}]\label{prop:fc_choice}
Let $G$ be a factor-critical graph, $v\in V(G)$,
and $M$ be a near-perfect matching that exposes $v$.
Then for any non-loop edge $e = vu \in E(G)$, 
there is a nice circuit $C$ of $G$
which is an $M$-ear relative to $v$ and 
contains $e$.
\end{proposition}
\begin{theorem}[\textmd{implicitly stated in~\cite{lovasz1972a}}]\label{thm:fc_nice}
Let $G$ be a factor-critical graph. 
For any nice factor-critical subgraph $G'$ of $G$,
$G/G'$ is factor-critical.
\end{theorem}
%
%
%
%
%
%
Let us denote the number of odd components (i.e. connected components 
with odd numbers of vertices) of a graph $G$ as $oc(G)$,
and the cardinality of a maximum matching of $G$ as $\nu(G)$.
It is known as the {\em  Berge formula}~\cite{lp1986} that 
for any graph $G$,
$|V(G)| - 2\nu(G) = \mathrm{max}\{ oc(G-X) - |X| : X\subseteq V(G) \}$.
%
A set of vertices that attains the maximum 
in the right side of the equation 
is called  a {\em  barrier}.

The \textit{canonical partition} is a decomposition  for elementary graphs
and plays a crucial role in matching theory.
First Kotzig introduced the canonical partition 
as a quotient set  of 
a certain equivalence relation~\cite{kotzig1959a, kotzig1959b, kotzig1960},
and later Lov\'asz redefined it from the point of view of barriers~\cite{lp1986}.
In fact, these are equivalent.
For an elementary graph $G$ and $u, v\in V(G)$,
we say 
$u\sim v$  if
$u= v$ or $G-u-v$ is not \matchable. 
\begin{theorem}[Kotzig~\cite{kotzig1959a, kotzig1959b, kotzig1960},
Lov\'asz~\cite{lp1986}] \label{thm:canonicalpartition}
Let $G$ be an elementary graph.
Then $\sim$ is an equivalence relation on $V(G)$ and 
the family of equivalence classes is  exactly the family of maximal barriers of $G$.
\end{theorem}
The family of equivalence classes of $\sim$
 is called the \textit{canonical partition} of $G$, and  denoted by $\mathcal{P}(G)$.
%
%
%
%
%
%

An {\em  ear-decomposition} of graph $G$ is a sequence of subgraphs 
$G_0, \subseteq, \cdots, \subseteq G_k = G$ such that 
 $G_0 = (\{r\}, \emptyset)$ for some $r\in V(G)$ 
and for each $i \ge 1$, $G_i$ is obtained from $G_{i-1}$ 
by adding an ear $P_i$ relative to $G_{i-1}$.
We sometimes regard  an ear-decomposition as 
 a family of ears $\mathcal{P} = \{P_1,\ldots, P_k\}$. 
An ear-decomposition is called {\em  odd} if 
any of its ears has an odd number of edges.
%
\begin{theorem}[Lov\'asz~\cite{lovasz1972a}]\label{thm:odd}
A graph  is factor-critical if and only if it has an odd ear-decomposition.
\end{theorem}
For a factor-critical graph $G$ and its near-perfect matching $M$,
we call an  ear-decomposition 
 \textit{alternating with respect to} $M$, or just   $M$-alternating, 
if each ear  
is an $M$-ear.  
%
%
\begin{proposition}[Lov\'asz~\cite{lovasz1972a}]~\label{prop:fc_alt}
Let $G$ be a factor-critical graph. Then for any near-perfect matching $M$ of $G$, 
there is an $M$-alternating  ear-decomposition of $G$.
\end{proposition}

\begin{proposition}\label{prop:allowed2circuit}
Let $G$ be a factorizable graph, and $M$ be a perfect matching of $G$. 
Then, for $e = xy \in E(G)\setminus M$, 
the followings are equivalent; 
\renewcommand{\labelenumi}{\theenumi}
\renewcommand{\labelenumi}{{\rm \theenumi}}
\renewcommand{\theenumi}{(\roman{enumi})}
\begin{enumerate}
\item 
$e$ is allowed in $G$. 
\item 
There is an $M$-alternating circuit containing $e$. 
\item 
There is an $M$-saturated path between $x$ and $y$. 
\end{enumerate}
\end{proposition}

\begin{proposition}\label{prop:ear2relative}
Let $G$ be a graph, $M$ be a matching of $G$, 
and $X\subseteq V(G)$ be such that $M_X$ is a perfect matching of $G[X]$. 
Let $P$ be a subgraph of $G$ that satisfies either of the followings;  
\renewcommand{\labelenumi}{\theenumi}
\renewcommand{\labelenumi}{{\rm \theenumi}}
\renewcommand{\theenumi}{(\roman{enumi})}
\begin{enumerate}
\item $P$ is an $M$-alternating circuit with $V(P)\cap X \neq \emptyset$, 
\item for some $u\in X$, $P$ is an $M$-ear relative to $\{u\}$, 
\item $P$ is an $M$-exposed path whose end vertices are in $X$, or 
\item $P$ is an $M$-saturated path whose end vertices are in $X$. 
\end{enumerate}
Then, 
each connected component of $P-E(G[X])$ is an $M$-ear relative to $X$. 
\end{proposition}

\section{A Partially Ordered Structure in  Factorizable Graphs}
Let $G$ be a factorizable graph. 
For $X\subseteq V(G)$ we call $X$ a {\em separating set} 
if for any $H\in \mathcal{G}(G)$, 
$V(H)\subseteq X$ or $V(H)\cap X = \emptyset$. 
The next property is easy to see by the definition. 
\begin{proposition}\label{prop:separating}
Let $G$ be a factorizable graph, and $X\subseteq V(G)$ with $X\neq \emptyset$. 
The following properties are equivalent; 
\renewcommand{\labelenumi}{\theenumi}
\renewcommand{\labelenumi}{{\rm \theenumi}}
\renewcommand{\theenumi}{(\roman{enumi})}
\begin{enumerate}
\item 
$X$ is separating. 
\item 
There exist $H_1,\ldots, H_k \in \mathcal{G}(G)$, where $k\ge 1$,  
such that $X = V(H_1)\dot{\cup}\cdots \dot{\cup} V(H_k)$. 
\item 
For any perfect matching $M$ of $G$, 
$\delta(X)\cap M = \emptyset$. 
\item 
For any perfect matching $M$ of $G$, 
$M_X$ is a perfect matching of $G[X]$. 
\end{enumerate}
\end{proposition}
%
%

Let $G_1, G_2\in\mathcal{G}(G)$. 
We say a separating set $X$ is a {\em critical-inducing set for} $G_1$ 
if $V(G_1)\subseteq X$ and $G[X]/G_1$ is a factor-critical graph. 
Moreover, 
we say $X$ is a {\em critical-inducing set for} $G_1$ {\em to} $G_2$ 
if $V(G_1)\cup V(G_2)\subseteq V(G)$ and $G[X]/G_1$ is a factor-critical graph. 
\begin{definition}
Let $G$ be a factorizable graph, and $G_1,G_2\in\mathcal{G}(G)$. 
We say $G_1\yield G_2$ if 
there is a critical-inducing set for $G_1$ to $G_2$. 
\end{definition}

\begin{lemma}\label{lem:path2base}
Let $G$ be a factorizable graph and $M$ be a perfect matching of $G$, 
and let $X\subseteq V(G)$ and $G_1\in\mathcal{G}(G)$. 
Then, $X$ is a critical-inducing set for $G_1$ if and only if 
for any $x\in X\setminus V(G_1)$ 
there exists $y\in V(G_1)$ such that 
there is an $M$-balanced path from $x$ to $y$ 
whose vertices except $y$ are in $X\setminus V(G_1)$. 
\end{lemma}
\begin{proof}
The claim is rather easy from Proposition~\ref{prop:path2root}; 
$X$ is a critical-inducing set for $G_1$ 
if and only if $G[X]/G_1$ is factor-critical. 
Note that $M_{X\setminus V(G_1)}$ forms a near-perfect matching of $G[X]/G_1$. 
Therefore, 
$G[X]/G_1$ is factor-critical if and only if 
for any $x\in X$ there is an $M$-balanced path 
from $x$ to the contracted vertex $g_1$ corresponding to $G_1$.  
Therefore, the claim follows. 
\qed
\end{proof}
\begin{proposition}\label{prop:nonpositive}
Let $G$ be an elementary graph and $M$ be a perfect matching of $G$.  
Then for any two vertices $u,v \in V(G)$
there is an $M$-saturated path between $u$ and $v$, 
or an $M$-\zero path from $u$ to $v$.
\end{proposition}
%
%
\begin{proof}
Without loss of generality we can assume $G$ is matching-covered, 
that is, every edge of $G$ is allowed.
Let $U_1\subseteq V(G)$ be the set of vertices that can be reached from $u$ by an  
$M$-saturated path,
and $U_2\subseteq V(G)$ be the set of vertices that can be reached from $u$ by an $M$-\zero path
but cannot be by any $M$-saturated paths. 
We are going to obtain the claim by showing $U:= U_1\dot{\cup} U_2 = V(G)$.
Suppose that it fails, namely that $U \subsetneq V(G)$.
Then there are $v\in U$ and $w\in V(G)\setminus U$ such that $vw\in E(G)$,
since $G$ is connected.
By the definition of $U$,
there is an $M$-saturated or \zero path $P$ from $u$ to $v$, 
which satisfies $V(P)\subseteq U$ 
since for each $z \in V(P)$ 
$uPz$ is an $M$-saturated or balaned path from $u$ to $z$. 
If $P$ is $M$-saturated, therefore, 
$P+vw$ is an $M$-balanced path from $u$ to $w$, 
which means $w\in U$, a contradiction.

Hence, hereafter we  assume $P$ is $M$-balanced, from $u$ to $v$.
Since $vw$ is defined to be allowed, 
there is an $M$-saturated path $Q$ between $v$ and $w$ 
by Proposition~\ref{prop:allowed2circuit}. 
Trace $P$ from $u$ and let $x$ be the first vertex we encounter 
that is in $Q$; 
such $x$ surely exists under the current hypotheses since $v\in V(P)\cap V(Q)$. 
\begin{cclaim}\label{claim:u2x}
$uPx$ is an $M$-balanced path. 
\end{cclaim}
\begin{proof}
Suppose the claim fails, 
which is equivalent to $uPx$ being an $M$-saturated path. 
Then, $x'\in V(uPx)$. 
On the other hand, 
since $Q$ is $M$-saturated, $x'\in V(Q)$. 
Therefore, $x'\in V(uPx)\cap V(Q)$,  
which means we counter $x'$ before $x$ if we trace $P$ from $u$, 
a contradiction. 
\qed
\end{proof}
\begin{cclaim}\label{claim:x2w}
$xPw$ is an $M$-saturated path between $x$ and $w$. 
\end{cclaim}
\begin{proof}
If $x = v$, $vPx$ is a trivial $M$-balanced path from $v$ to $x$. 
Even if $x \neq v$, so is it since $x$ is matched by $M\cap E(P)$. 
Anyway, whether $x = v$ or not, $vPx$ is an $M$-balanced path 
from $v$ to $x$. 
Therefore, together with $vPw$ being an $M$-saturated path, 
$xPw$ is an $M$-balanced path from $x$ to $w$. 
\qed
\end{proof}
By Claims~\ref{claim:u2x} and \ref{claim:x2w}, 
$uPx + xQw$ is an $M$-saturated path between $u$ and $w$, 
since $V(uPx)\cap V(xQw) = \{ x\}$ by the definition of $x$. 
Hence, $w\in U$, a contradiction, 
and we obtain $U = V(G)$, which completes the proof. 
\qed
\end{proof}  

\begin{proposition}\label{prop:adjoin-ear}
Let $G$ be a factorizable graph and $M$ be a perfect matching of $G$. 
Let $X\subseteq V(G)$,  and $H \in\mathcal{G}(G)$ be such that 
there is an $M$-ear $P$ relative to $X$ and through $H$, 
whose end vertices are $u, v\in V(G_1)$. 
Let $Y := V(H) \cup V(P)\setminus \{u, v\}$. 
Then, for any $x\in Y$, 
\renewcommand{\labelenumi}{\theenumi}
\renewcommand{\labelenumi}{{\rm \theenumi}}
\renewcommand{\theenumi}{(\roman{enumi})} 
\begin{enumerate}
\item \label{item:int}
there exists an internal vertex $y$ of $P$ such that 
there is an $M$-balanced path $Q$ from $x$ to  $y$ with  
$V(Q)\subseteq Y$ and $V(Q)\cap V(P) = \{y\}$, and 
\item \label{item:end}
for $w$ identical to either $u$ or $v$, 
$Q + yPw$ is an $M$-balanced path from 
$x$ to $w$, 
whose vertices except $w$ are contained in $Y$.  
\end{enumerate}
\end{proposition}
\begin{proof}
If $x\in V(P)\setminus \{u, v\}$, the claims are obvious. 
Let $x\in V(H)\setminus V(P)$. 
Then, by Proposition~\ref{prop:nonpositive}, 
for an arbitrarily chosen $z\in V(P)\cap V(H)$,  
there is an $M$-saturated or balanced path $R$ from $x$ to $z$ 
with $V(R)\subseteq V(H)$. 
Trace $R$ from $x$ and let $y$ be the first vertex we encounter that is in $V(P)$. 
Then, $xRy$ gives a desired path in \ref{item:int}, 
and 
$Q:= xRy + yPw$, where $w$ is either $u$ or $v$, 
gives one  in \ref{item:end}. 
Therefore, we are done. 
\qed
\end{proof}
%
%

Let $G$ be a factorizable graph and $M$ be a perfect matching of $G$.
We call a sequence of factor-components
$S:= (H_0,\ldots, H_k)$, where $k\ge 0$ 
and $H_i\in\mathcal{G}(G)$ for each $i = 0, \ldots, k$,   an 
\textit{$M$-ear sequence,  from $H_0$ to $H_k$, }
if $k = 0$ or otherwise 
\renewcommand{\labelenumi}{\theenumi}
\renewcommand{\labelenumi}{{\rm \theenumi}}
\renewcommand{\theenumi}{(\roman{enumi})}
\begin{enumerate}
\item  
for any $i, j\in\{0,\ldots, k\}$, $i\neq j$ yields $H_i\neq H_j$,  and 
\item 
for each $i = 1,\ldots, k$
there is an $M$-ear $P_i$ relative to $H_{i-1}$ and through $H_i$. 
\end{enumerate}
We call $k$ the {\em length} of $S$. 
If $k\ge 1$, we call the sequence of $M$-ears 
$P:= (P_1, \ldots , P_k)$ {\em associated with} $S$. 
If $k = 0$, 
an empty sequence, $P:= ()$, 
is defined to be the $M$-ears associated with $S$, for convenience.  
 
For $S$ and $P$,  
we define the {\em sequence union} of $S$ and $P$ as 
$\su{S}{P} := \bigcup_{i=1}^{k} V(H_i) \cup \bigcup_{i=1}^{k} V(P_i) \setminus V(H_0)$, 
if $k\ge 1$.  
If $k = 0$, $\su{S}{P} := \emptyset$. 

Given $S$ and $P$, 
for any $i, j$ with $0 \le i \le j \le k$, 
the subsequence 
$(H_i, \ldots,  H_j)$ is an $M$-ear sequence, from $H_i$ to $H_j$, 
and we denote it as $S[i,j]$.  
Additionally, if $i < j$, 
$(P_i,\ldots, P_j)$ is a sequence of $M$-ears associated with $S[i,j]$, 
and we denote it $P[i,j]$. 
If $i = j$, $P = ()$ is  associated with $S[i,j]$, 
and it is also denoted as $P[i, j]$.
We denote $S[0, j] =: S^j$, and $P[0, j] =: P^j$. 
%
%


Let $G$ be a factorizable graph, 
and $M$ be a perfect matching of $G$. 
Let $G_1, G_2\in\mathcal{G}(G)$, 
and let 
$S:= (G_1 = H_0, \ldots, H_k = G_2)$, where $k\ge 0$, 
 be an $M$-ear sequence from $G_1$ to $G_2$,  
associated with $M$-ears $P$. 
Let us define in the following three properties 
for $S$ and $P$: 
\begin{description}
\item{\bf{D1}($S,P$):}
If $k\ge 2$, then by letting $P = (P_1, \ldots, P_k)$, for each $i = 2,\ldots, k$, 
$V(P_i)$ is disjoint from $V(H_0)$. 
\item{\bf{D2}($S, P$):}
If $k\ge 1$,  
by letting $P = (P_1, \ldots, P_k)$, 
for each $i = 1,\ldots, k$, 
for any $x\in \su{S^i}{P^i}$ 
there exists an internal vertex $y$ of $P_1$ such that 
there is an $M$-balanced path $Q$ from $x$ to $y$ 
with $V(Q)\subseteq \su{S^i}{P^i}$ and $V(Q)\cap V(P_1) = \{y\}$. 
\item{\bf{D3}($S, P$):}
If $k\ge 1$, 
by letting $P = (P_1, \ldots, P_k)$,  
for each $i = 1,\ldots, k$, 
for any $x\in \su{S^i}{P^i}$,   
for $w$ which equals either of the end vertices of $P_1$, 
there is an $M$-balanced path $R$ from $x$ to $w$ 
such that $V(P_1)\setminus \{w\} \subseteq \su{S^i}{P^i}$. 
\end{description}
\begin{remark}
By their definitions, 
if $k = 0$, then $S$ and $P$ trivially satisfy 
D1, D2 and D3. 
\end{remark}
\begin{remark}
D1, D2 and D3 are closed with respect to the substructures; 
if  $S$ and $P$ satisfies D1, D2 and D3, 
then for any $i = 0,\ldots, k$, 
 so does $S^i$ and $P^i$. 
\end{remark} 
%
%

\begin{proposition}\label{prop:su-matching}
Let $G$ be a factorizable graph and $M$ be a perfect matching of $G$. 
Let $S$ be an $M$-ear sequence, and $P$ be a sequence of $M$-ears associated with $S$. 
Then, 
$M_{\su{S}{P}}$ is a perfect matching of $G[\su{S}{P}]$. 
\end{proposition}
\begin{proof}
If the length $k$ of $S$ equals zero, 
the claim is trivially true. 
Let $k\ge 1$, and let 
$S =: (H_0, \ldots, H_k)$ 
and $P =: (P_1, \ldots, P_k)$. 
Of course, 
$X := V(H_0)\dot{\cup}\cdots \dot{\cup} V(H_k)$ 
has a perfect matching $M_X$.  
For each $P_i$, 
the end vertices of $P_i$ are in $X$ 
and any other vertex is covered by $M_{P_i}$. 
Therefore, 
$M$ contains a perfect matching of 
$Y:= X\cup V(P_1)\cup\cdots \cup V(P_k)$. 
Accordingly, $\su{S}{P} = Y\setminus V(H_0)$ is covered by $M_{\su{S}{P}}$. 
\qed
\end{proof}

\begin{lemma}\label{lem:fc2bigger}
Let $G$ be a factorizable graph and $M$ be a perfect matching of $G$. 
Let $G_1\in\mathcal{G}(G)$ and $X\subseteq V(G)$ be a critical-inducing set for $G_1$. 
Suppose there exists an $M$-ear $P$ relative to $X$, 
whose end vertices are $u, v\in V(G)$, 
and let $I_1,\ldots, I_s\in \mathcal{G}(G)$, where $s\ge 1$,  be 
the factor-components that have common vertices with the internal vertices of $P$. 
Then, $X\cup \bigcup_{i=1}^s V(I_i)$ is also a critical-inducing set for $G_1$. 
\end{lemma}
\begin{proof}
We prove the claim by Lemma~\ref{lem:path2base}; 
let $Y:= \bigcup_{i=1}^s V(I_i)$.  
By Lemma~\ref{lem:path2base}, 
\begin{cclaim}\label{claim:inx} 
for any $x\in X$ there exists $z\in V(G_1)$ such that 
there is an $M$-balanced path $Q_x$ from $x$ to $z$ with 
$V(Q_x)\subseteq X$ and $V(Q_x)\cap V(G_1) = \{z\}$. 
\end{cclaim}
\begin{cclaim}\label{claim:iny}
For any $y\in Y$ there exists $z\in V(G_1)$ such that there exists an 
$M$-balanced path $Q_y$ from $y$ to $x$ 
with $V(Q_y)\subseteq X$ and $V(Q_y)\cap V(G_1) = \{y\}$.  
\end{cclaim}
\begin{proof} 
Let $i\in\{1,\ldots, s\}$ be such that $y\in V(I_i)$. 
By applying Proposition~\ref{prop:adjoin-ear} to 
$X$, $I_i$ and $P$, 
for $w$ which equals either $u$ or $v$, 
there is an $M$-balanced path $R$ from $y$ to $w$ 
such that $V(R)\setminus \{w\} \subseteq Y$. 
Therefore, 
$P + Q_w$ gives a desired path. 
\qed
\end{proof}
Apparently by the definition 
$X\cup Y$ is a separating set, 
therefore with Claims~\ref{claim:inx} and \ref{claim:iny}
we can conclude that 
$X\cup Y$ is a critical-inducing set for $G_1$, 
by Lemma~\ref{lem:path2base}. 
\qed
\end{proof}

\begin{theorem}\label{thm:equivalentdefinition}
Let $G$ be a factorizable graph, 
$M$ be a perfect matching of $G$, and 
 $G_1, G_2\in\mathcal{G}(G)$.
Then, $G_1\yield G_2$ if and only if 
there exists 
an $M$-ear sequence from $G_1$ to $G_2$.
\end{theorem}
\begin{proof}
We first prove the sufficiency. 
Let $G_1\yield G_2$ and $X\subseteq V(G)$ be a critical-inducing set for $G_1$ to $G_2$.
Let us define the following three properties for $Y\subseteq X$: 
\begin{description}
\item[C1($Y$):] \label{one} $Y$ is a critical-inducing set for $G_1$, and 
\item[C2($Y$):] \label{four} for each $H\in\mathcal{G}(G)$ with $V(H)\subseteq Y$, there is an $M$-ear sequence 
from $G_1$ to $H$. 
\end{description}
Let $X'$ be a maximal subset of $X$ satisfying C1 and C2. 
Note that $X'\neq \emptyset$ because $V(G_1)$ satisfies 
C1 and C2.
We are going to prove the sufficiency by showing 
that $X' = X$. 
Suppose it fails, that is, $X'\subsetneq X$. 
Then, 
\begin{cclaim}
there is an $M$-ear $P$ relative to $X'$ such that $V(P)\subseteq X$. 
\end{cclaim}
\begin{proof}
$G[X]/G_1$ is factor-critical and $G[X']/G_1$ is a nice factor-critical subgraph of $G[X]/G_1$
by Proposition~\ref{prop:separating}. 
Therefore, $G[X]/X'$ is factor-critical by Theorem~\ref{thm:fc_nice} 
and $M_{X\setminus X'}$ forms a near-perfect matching of $G[X]/X'$ 
exposing only the contracted vertex $x'$ corresponding to $X'$. 
By Proposition~\ref{prop:fc_choice}, 
in $G[X]/X'$ there is an $M$-ear $P$ relative to $x'$, 
and  in $G$  it corresponds to an $M$-ear relative to $X'$ with $V(P)\subseteq X$. 
Thus, the claim follows. 
\qed
\end{proof} 
Let $u, v\in X'$ be the end vertices of $P$. 
Let  $I_1,\ldots, I_s \in \mathcal{G}(G)$ be the factor-components 
 that have common vertices with internal vertices of $P$. 
We are going to prove that  
$X'' := X'\cup \bigcup_{i=1}^{s} V(I_i)$ 
satisfies C1 and C2. 
%
%
%
%
\begin{cclaim}\label{claim:c2}
$X''$ satisfies C2. 
\end{cclaim}
\begin{proof}
By Lemma~\ref{lem:path2base}, 
there exists an $M$-balanced path $Q_u$ (resp. $Q_v$) 
from $u$ (resp. $v$) to a vertex of $V(G_1)$, 
which is contained in $X$ and whose vertices except the end vertex in $V(G_1)$ 
are disjoint from $V(G_1)$. 
Trace $Q_u$ from $u$ and let $r_u$ be the first vertex we encounter 
that is contained in a factor-component $I_0$ which 
has common vertices also with $Q_v$; 
such $I_0$ surely exists since both $Q_u$ and $Q_v$ have some vertices in $G_1$. 
Trace $Q_v$ from $v$ and let $r_v$ be the first vertex we encounter 
that is in $V(I_0)$. 
For each $w\in \{u, v\}$,  
$wQ_wr_w$ 
is an $M$-balanced path from $w$ to $r_w$ 
such that 
$V(wQ_wr_w)\subseteq X'$ and $V(wQ_wr_w)\cap V(I_0)=\{r_w\}$,  
and it holds that 
$V(uQ_ur_u)\cap V(vQ_vr_v) \setminus \{r_u, r_v\} = \emptyset$. 
Therefore, 
$uQ_ur_u + P + vQ_vr_v$ is an $M$-ear relative to $I_0$ 
and through every $I_1,\ldots, I_s$. 
By the definition of $X'$, 
there is an $M$-ear sequence from $G_1$ to $I_0$. 
Therefore, by adding subsequence $(I_0, I_i)$ to it, 
we obtain an $M$-ear sequence from $G_1$ to $I_i$, 
for each $i = 1,\ldots, s$. 
Thus, we  obtain the claim. 
\qed
\end{proof}

%
%
\begin{cclaim}\label{claim:c1}
$X''$ satisfies C1. 
\end{cclaim}
\begin{proof}
This is immediate by Lemma~\ref{lem:fc2bigger}. 
\qed
\end{proof}
%
With Claims~\ref{claim:c2} and \ref{claim:c1},  
$X''$ contradicts the maximality of $X'$. 
Therefore, we obtain $X' = X$, 
accordingly the sufficiency part of the claim follows. 

From now on we prove the necessity.
Let $(G_1 = H_0,\ldots, H_k = G_2)$, where $k\ge 0$, 
 be the $M$-ear sequence from $G_1$ to $G_2$
We are going to prove that
there is a critical-inducing set for $G_1$  to $G_2$. 
We proceed by induction on $k$.
For the case $k=0$, that is, $G_1 = G_2$, 
the claim apparently holds by taking $V(G_1)$.

Let $k>0$ and suppose the claim holds for $k-1$.  
By the induction hypothesis, for the $M$-ear subsequence $(H_0, \ldots, H_{k-1})$, 
there is a critical-inducing set $X'$ for $H_0$ to $H_{k-1}$. 
\begin{cclaim}
There is an $M$-ear $P$ relative to $X'$ and through $H_k$. 
\end{cclaim}
\begin{proof}
Let $P_k$ the associated $M$-ear relative to $H_{k-1}$ and through $H_k$.  
By Proposition~\ref{prop:ear2relative} 
each connected component $P - E(G[X'])$ is an $M$-ear relative to $X'$, 
and one of them, which we call $P$,  is through $H_k$. 
Therefore, the claim follows. 
\qed
\end{proof}
Let $I_1,\ldots, I_s \in \mathcal{G}(G)$, where $s\ge 1$,  be the factor-components
that have common vertices with the  internal vertices of $P$, 
and let $Y:= \bigcup_{i=1}^{s} V(I_i)$. 
%
%
Then, by applying Lemma~\ref{lem:fc2bigger} 
to the critical-inducing set $X'$ for $G_1$ and the $M$-ear $P$, 
we obtain that $X'\cup Y$ is a critical-inducing set for $G_1$ to $H_k$. 
This completes the proof. 
\qed
\end{proof}

\begin{lemma}\label{lem:inductive}
Let $G$ be a factorizable graph, 
and $M$ be a perfect matching.  
Let $S:= (H_0,\ldots, H_k)$, where $k\ge 1$, 
 be an $M$-ear sequence, 
 associated with $M$-ears $P:= (P_1,\ldots, P_k)$. 
Suppose  $S^i$ and $P^i$  satisfy D1, D2, and D3 for each $i = 0,\ldots, k-1$, 
and  $S$ and $P$ satisfy D1. 
Then, $S$ and $P$  also satisfy D2 and D3. 
\end{lemma}
\begin{proof}
If $k =  1$, then by applying Proposition~\ref{prop:adjoin-ear} to 
$V(H_0)$, $P_1$, and $H_1$, it holds that $S$ and $P$ satisfy 
D1, D2 and D3. 

Hence hereafter let $k\ge 2$. 
First note that each connected component of $P_k- E(G[\su{S^{k-1}}{P^{k-1}}])$ 
is an $M$-ear relative to $\su{S^{k-1}}{P^{k-1}}$ by Proposition~\ref{prop:ear2relative}, 
and is disjoint from $V(H_0)$ since $P_k$ is.  

Take $x\in \su{S}{P}\setminus \su{S^{k-1}}{P^{k-1}}$ arbitrarily, 
 and let $P_k^x$ be a connected component of $P_k- E(G[\su{S^{k-1}}{P^{k-1}}])$  
 such that $x$ is an internal vertex of $P_k^x$ if $x\in V(P)$, 
 or one through $H_k$ if $x\in V(H_k)\setminus V(P)$.  
\begin{cclaim}\label{claim:x2y}
There exists $y\in \su{S^{k-1}}{P^{k-1}}$ such that 
there exists an $M$-balanced path $Q$ from $x$ to $y$ 
whose vertices except $y$ are contained in $\su{S}{P}\setminus \su{S^{k-1}}{P^{k-1}}$. 
\end{cclaim}
\begin{proof}
By applying Proposition~\ref{prop:adjoin-ear} 
to $\su{S^{k-1}}{P^{k-1}}$,  
$P_k^x$, and $H_k$ (if $x\in V(H_k)$),  
we obtain an internal vertex $y$ of $P_1$ 
and an $M$-balanced path $Q$ from $x$ to $y$ 
with $V(Q)\setminus \{y\}\subseteq V(H_k)\cup V(P_k^x)\setminus \su{S^{k-1}}{P^{k-1}}$. 
Since $P_k$ is disjoint from $V(H_0)$, 
we can see $V(H_k)\cup V(P_k^x) \subseteq \su{S}{P}$. 
Therefore, $V(Q)\setminus \{y\} \subseteq \su{S}{P}\setminus \su{S^{k-1}}{P^{k-1}}$, 
and the claim follows.  
\qed
\end{proof}

\begin{cclaim}
$S$ and $P$ satisfy D2. 
\end{cclaim}
\begin{proof}
By the hypothesis on $S^{k-1}$ and $P^{k-1}$ 
there exists an internal vertex $z$ of $P_1$ such that 
there is an $M$-balanced path $R$ from $y$ to $z$ 
with $V(R)\subseteq \su{S^{k-1}}{P^{k-1}}$ and $V(R)\cap V(P_1) = \{z\}$. 
Therefore, by Claim~\ref{claim:x2y}, $Q + R$ is an $M$-balanced path from $x$ to $z$, 
whose veritices are contained in $\su{S}{P}$ 
and disjoint from $P_1$ except $z$.  

Since $x$ is chosen arbitrarily from $\su{S}{P}\setminus \su{S^{k-1}}{P^{k-1}}$, 
we obtain that $S$ and $P$ satisfy D2. 
\qed
\end{proof}
By similar arguments, we can say that 
$S$ and $P$ satisfy D3 too, 
and the claim follows. 
%
%
%
\end{proof}

\begin{proposition}\label{prop:ear}
Let $G$ be a factorizable graph and $M$ be a perfect matching. 
Let $G_1, G_2\in\mathcal{G}(G)$ be such that $G_1\yield G_2$, 
and let $k\ge 0$ be the length of the shortest $M$-ear sequence from $G_1$ to $G_2$. 
Then, there exists an $M$-ear sequence 
$S$ of shortest length, 
and $M$-ears $P$ associated with $S$ 
such that D1($S$, $P$), D2($S$, $P$), and D3($S$,$P$) hold. 
\end{proposition}
\begin{proof}
%
We proceed by induction on $k$. 
If $k = 0$, the claim is trivially true. 
If $k = 1$, 
for any shortest $M$-ear sequence $S = (H_0 = G_1, H_1 = G_2)$ from $G_1$ to $G_2$ 
and associated $M$-ears $P = (P_1)$, 
D1($S$, $P$) trivially holds by the definition of D1,   
and moreover D2($S$, $P$) and D3($S$, $P$) also hold  
by applying Proposition~\ref{prop:adjoin-ear} to 
$V(H_0)$, $P_1$, and $H_k$. 

Let $k\ge 2$, and suppose the claim is true for 
any two factor-components $G_1', G_2'\in\mathcal{G}(G)$ such that 
the length of the shortest $M$-ear sequence from $G_1'$ to $G_2'$, 
is $1,\ldots, k-1$.  

Take arbitrarily an $M$-ear sequence $S = (G_1 = H_0, \ldots, H_k = G_2)$ 
from $G_1$ to $G_2$  of shortest length, 
and $M$-ears $P = (P_1,\ldots, P_k)$ associated with it. 
Let $u_1, v_1$ be the end vertices of $P_1$. 
\begin{cclaim}
Without loss of generality we can assume that 
$S$ and $P$ are chosen so that 
for each $i = 1, \ldots, k-1$, 
$S^i$ and $P^i$ satisfy D1, D2, and D3. 
\end{cclaim}
\begin{proof}
By the induction hypothesis, 
there exist an $M$-ear sequence from $H_0$ to $H_{k-1}$, 
which is of shortest length,   
and $M$-ears associated with it 
which satisfy D1, D2, and D3; 
note that its length is $k-1$. 
Without loss of generality, 
we can assume $S^{k-1}$ and $P^{k-1}$ coincides to them. 
Since the conditions D1, D2, and D3 are closed with 
substructures, 
the claim follows. 
\qed
\end{proof}
%
%
%
%
%
If $P_k$ is disjoint from $V(H_0)$, 
namely if D1($S$, $P$) holds, 
then by Lemma~\ref{lem:inductive}, $S$ and $P$ also satisfy D2 and D3, 
and the claim follows. 

Hence hereafter 
suppose that might fail i.e. $P_k$ might not be disjoint from $V(H_0)$. 
By Proposition~\ref{prop:ear2relative}, 
each connected component of $P_k - E(G[\su{S^{k-1}}{P^{k-1}}])$ 
is an $M$-ear relative to $\su{S^{k-1}}{P^{k-1}}$. 
Take one of them $Q$ arbitrarily 
that has common vertices with $H_k$. 
 
Take $x\in V(Q)\cap V(H_k)$ arbitrarily, 
and let $u, v$ be the end vertices of $Q$. 
Trace $xQu$ from $x$ and let $y$ be the first vertex we encounter 
that is in $V(H_0)\cup \{u\}$. 
On the other hand, 
trace $xQv$ from $x$ and let $z$ be the first vertex we encounter 
that is in $V(H_0)\cap \{v\}$. 
Then, 
\begin{cclaim}
$yQz$ is an $M$-exposed path, 
whose internal vertices contains $x\in V(H_k)$, 
and whose vertices except the end vertices $y$ and $z$ 
are disjoint from $V(H_0)\cup \su{S^{k-1}}{P^{k-1}}$. 
\end{cclaim} 
\begin{cclaim}\label{claim:q2d1}
$Q$ is disjoint from $V(H_0)$. 
\end{cclaim}
\begin{proof}
We are going to prove $y = u$ and $z = v$; 
First suppose the case where $y, z\in V(H_0)$. 
Then, $yQz$ is an $M$-ear relative to $H_0$ and through $H_k$, 
which means $(H_0, H_k)$ forms an $M$-ear sequence of length one, 
contradicting the definition of $k$, since $k\ge 2$. 

Second suppose the case where $y\in V(H_0)$ and $z = v$. 
Since $S^{k-1}$ and $P^{k-1}$ satisfy D3, 
for either $w\in \{u_1, v_1\}$ 
there is an $M$-balanced path $R$ from $z$ to $w$ 
such that $V(R)\setminus \{w\}\subseteq \su{S^{k-1}}{P^{k-1}}$. 
Therefore, $yQz + R$ is an $M$-ear relative to $H_0$ and through $H_k$, 
again letting $(H_0, H_k)$ be an $M$-ear sequence, 
a contradiction.  

In the third case where $y=u$ and $z\in V(H_0)$, 
by symmetric arguments we are again lead to a contradiction. 

Therefore, we obtain that $y = u$ and $z = v$, 
which is equivalent to $Q$ being disjoint from $V(H_0)$. 
\qed
\end{proof}
Since $S^{k-1}$ and $P^{k-1}$ satisfy D3, 
for each $\alpha \in \{u, v\}$ 
there is an $M$-balanced path $Q_\alpha$ from $\alpha$ to $r_\alpha$, 
where $r_\alpha$ equals either $u_1$ or $v_1$, 
such that $V(Q_\alpha) \setminus \{r_\alpha\} \subseteq \su{S^{k-1}}{P^{k-1}}$. 
Trace $Q_u$ from $u$ and let $s$ be the first vertex we encounter 
that is contained in a factor-component, say $I\in\mathcal{G}(G)$, 
which has common vertices also with $V(Q_v)$; 
such $I$ surely exists since both $Q_u$ and $Q_v$ have vertices in $H_0$. 
Trace $Q_v$ from $v$ and let $t$ be the first vertex we encounter 
that is in $V(I)$. 
\begin{cclaim}\label{claim:noth0}
$I\neq H_0$. Accordingly, $V(Q_u)\cup V(Q_v)\subseteq \su{S^{k-1}}{P^{k-1}}$. 
\end{cclaim}
\begin{proof}
$\su{S^{k-1}}{P^{k-1}}\cap V(H_0) = \emptyset$, 
and for each $\alpha \in \{u, v\}$, 
$V(Q_\alpha)\setminus \{r_\alpha\} \subseteq \su{S^{k-1}}{P^{k-1}}$. 
Therefore, 
$I = H_0$ only if $V(Q_u)\cap V(Q_v) = \emptyset$ or 
$V(Q_u)\cap V(Q_v) = \{r_u\} = \{r_v\}$. 
Then, 
$Q_u + Q + Q_v$ forms an $M$-ear relative to $H_0$ and through $H_k$, 
letting $(H_0, H_k)$ be an $M$-ear sequence of length one, 
a contradiction. 
\qed
\end{proof}
\begin{cclaim}
Each connected component of $uQ_us + Q + vQ_zt - E(I)$ 
is an $M$-ear relative to $I$, one of which, say $\hat{Q}$,  is through $H_k$. 
\end{cclaim}
\begin{proof}
$uQ_us$ and $vQ_zt$ are $M$-balanced paths respectively from $u$ to $s$ and 
from $v$ to $t$, 
and they are disjoint if $s \neq t$, 
or have only one common vertex $s = t$ if $s = t$.  
Additionally, they are both contained in $\su{S^{k-1}}{P^{k-1}}$ 
by Claim~\ref{claim:noth0}, 
while $V(Q)\cap \su{S^{k-1}}{P^{k-1}} = \{u, v\}$. 
Therefore, 
$uQ_us + Q + vQ_zt$ forms an $M$-exposed path between $s$ and $t$ 
if $s\neq t$, or an $M$-ear relative to $\{s\} = \{t\}$ if $s = t$, 
in both cases having internal vertices contained in $H_k$, 
since $Q$ does. 
Hence, by Proposition~\ref{prop:ear2relative}, 
the claim follows. 
\qed
\end{proof}
%
%

By the arguments up till now,  
$I$ has some vertices in $\su{S^{k-1}}{P^{k-1}}$. 
Hence, 
$I$ equals either of $H_1,\ldots, H_{k-1}$ or otherwise 
it just has common vertices other than $u_1$ or $v_1$, 
 with either of $P_1,  \ldots, P_{k-1}$. 
\begin{cclaim}\label{claim:i2elem}
If $I$ equals either of $H_1,\ldots, H_{k-1}$, 
then $I = H_{k-1}$. 
\end{cclaim}
\begin{proof}
If $k = 2$, the claim is trivially true. 
Let $k \ge 3$ and 
suppose the claim fails, that is, 
$I = H_i$ for $i\in \{1, \ldots, k-2\}$. 
Then, $(H_0, \ldots, H_i = I, H_k)$ 
forms an $M$-ear sequence from $H_0$ to $H_k$, 
associated with $(P_1,\ldots, P_i, \hat{Q})$, 
and of length $i + 1 \le k-1$. 
This contradicts the definition of $k$, 
therefore we have the claim. 
\qed
\end{proof}
\begin{cclaim}\label{claim:2d1}
$\hat{Q}$ is disjoint from $V(H_0)$. 
\end{cclaim}
\begin{proof}
By Claim~\ref{claim:q2d1}, $Q$ is disjoint from $V(H_0)$, 
and by Claim~\ref{claim:noth0}, 
$Q_u$ and $Q_v$ are both disjoint from $V(H_0)$. 
Therefore, 
$\hat{Q} = Q_u + Q + Q_v$ is also disjoint from $V(H_0)$. 
\qed
\end{proof}
Therefore, with Claims~\ref{claim:i2elem} and \ref{claim:2d1},  in the above case, 
namely where $I = H_{k-1}$, 
$S = (H_0, \ldots, H_k)$ is an $M$-ear sequence, 
which can be regarded as being associated by 
$M$-ears $P' := (P_1,\ldots, P_{k-1}, \hat{Q})$. 
Since $S$ and $P'$ satisfy D1 by Claim~\ref{claim:2d1}, 
we have that they satisfy also D2 and D3, 
by Lemma~\ref{lem:inductive}. 
Hence we are done for this case. 
\begin{cclaim}\label{claim:i2ear}
If $I$ is distinct from any of $H_1,\ldots, H_{k-1}$, 
then $P_{k-1}$, an $M$-ear relative to $H_{k-2}$,  is through $I$. 
\end{cclaim}
\begin{proof}
If $k = 2$, the claim apparently follows. 
Let $k\ge 3$ and 
suppose $I$ has common vertices with $P_i$ with $i\in \{1,\ldots, k-2\}$. 
Namely, 
$P_i$ is an $M$-ear relative to $H_{i-1}$ and through $I$. 
Hence, 
$(H_0, \ldots, H_{i-1}, I, H_k)$ is an $M$-ear sequence from $H_0$ to $H_k$, 
of length $i+1\le k-1$, 
associated with $M$-ears $(P_1,\ldots, P_i, \hat{Q})$.  
This contradicts the definition of $k$.
Therefore we can conclude that $i = k-1$, 
and the claim follows. 
\qed
\end{proof}
Therefore, in this case, by Claim~\ref{claim:i2ear}, 
$\tilde{S}:= (H_0, \ldots, H_{k-2}, I, H_k)$ 
is an $M$-ear sequence  associated with  
$\tilde{P} = (P_1,\ldots, P_{k-1}, \hat{Q})$.

\begin{itemize}
\item 
$\tilde{S}^{k-2}$ and $\tilde{P}^{k-2}$ satisfy D1, D2, and D3, 
since $\tilde{S}^{k-2} = S^{k-2}$ and $\tilde{P}^{k-2} = P^{k-2}$, 
and 
\item $\tilde{S}^{k-1}$ and $\tilde{P}^{k-1}$ satisfy D1, 
since $(k-1)$-th elements of $P$ and $\tilde{P}$ are identical. 
\end{itemize}
Therefore, by Lemma~\ref{lem:inductive}, 
$\tilde{S}^{k-1}$ and $\tilde{P}^{k-1}$ also satisfy D2 and D3. 
Moreover, by Claim~\ref{claim:2d1}, 
with Lemma~\ref{lem:inductive} again applied to 
$\tilde{S}$ and $\tilde{P}$, 
we obtain, with Claim~\ref{claim:2d1}, 
 that $\tilde{S}$ and $\tilde{P}$ also satisfy D1, D2, and D3. 
This complets the proof. 
%
\qed
\end{proof}

\begin{theorem}\label{thm:order}
$\yield$ is a partial order.
\end{theorem}
\begin{proof}
The reflexivity is obvious from the definition.
The transitivity  obviously follows from Theorem~\ref{thm:equivalentdefinition}.
Hence, we will prove the antisymmetry. 
Let $G_1, G_2\in \mathcal{G}(G)$ be such that $G_1\yield G_2$ and $G_2\yield G_1$. 
Suppose that the antisymmetry fails, that is, that $G_1 \neq G_2$. 
Let $M$ be a perfect matching of $G$. 
By Proposition~\ref{prop:ear}, 
there exists an $M$-ear sequence from $G_1$ to $G_2$, say 
$S:= ( G_1 = H_0,\ldots, H_k = G_2)$, where $k \ge 1$, 
and associated $M$-ears $P: = (P_1,\ldots, P_k)$ 
which  D1, D2 and D3. 
Let $u_1$ and $v_1$ be the end vertices of $P_1$.

By Lemma~\ref{lem:path2base}
there exists $w\in V(G_2)$ such that   
there is an $M$-balanced path $Q$ from $u_1$ to $w$. 
Trace $Q$ from $u_1$ and let $x$ be the first vertex we encounter 
that is in $(\su{S}{P}\cup \{v_1\})\setminus \{u_1\}$; 
such a vertex surely exists since $V(G_2)\subseteq \su{S}{P}$. 
\begin{cclaim}
Without loss of generality we can assume that 
$x\neq v_1$, namely, 
$x\in \su{S}{P}$ and $u_1Qx$ is disjoint from $v_1$. 
\end{cclaim}
\begin{proof}
Suppose the claim fails, that is, $x = v_1$. 
Then, $u_1\neq v_1$. 
If $u_1Qv_1$ is an $M$-saturated path, 
then $P_1 + u_1Qv_1$ forms an $M$-alternating circuit, containing non-allowed edges, 
a contradiction. 
Otherwise, namely if $u_1Qv_1$ is an $M$-balanced path from $u_1$ to $v_1$,  
then $v_1Qw$ is an $M$-balanced path from $v_1$ to $w$. 
Now redefine $x$  as the first vertex we encounter that is in $\su{S}{P}$ 
if we trace $v_1Qw$ 
from $v_1$. 
Then, $v_1Qx$ is an $M$-balanced path from $v_1$ to $x$ which is disjoint from $u_1$. 
Therefore, by changing the roles of $u_1$ and $v_1$, 
without loss of generality, 
we obtain the claim.  
\qed
\end{proof} 
Therefore, hereafter let $x \in  \su{S}{P}$, noting that  
$u_1Qx$ is an $M$-balanced path from $u_1$ to $x$.  
Since $x\in \su{S}{P}$, by Proposition~\ref{prop:adjoin-ear}
there is an $M$-balanced path $R$ from $x$ to an internal vertex of $P_1$, say $y$, 
such that $V(R)\subseteq \su{S}{P}$ and $V(R)\cap V(P_1) = \{y\}$.

If $u_1P_1y$ has an even number of edges, 
$u_1Qx + xRy + yP_1u_1$ is an $M$-alternating circuit containing non-allowed edges, 
a contradiction.

Hence hereafter we assume $u_1P_1y$ has an odd number of edges.
By Proposition~\ref{prop:nonpositive},
there is an $M$-saturated or balanced path $L$ from $v_1$ to $u_1$
which is contained in $G_1$. 
Trace $L$ from $v_1$ and let $w$ be the first vertex on $u_1Qx$; 
note that $L$ is disjoint from $\su{S}{P}$ 
since $V(L)\subseteq V(H_0)$ and $\su{S}{P}$ is disjoint from $V(H_0)$.  
If $u_1Qw$ has an odd number of edges, then
$wQu_1 + P_1 + v_1Lw$ is an $M$-alternating circuit, a contradiction.
If  $u_1Qw$ has an even  number of edges,
then $v_1Lw+ wQx + xRy + yP_1u_1$ is an $M$-alternating circuit,
which is also a contradiction.
Thus we get $G_1 = G_2$, and  the claim follows.
\qed
\end{proof}

\section{A Generalization of the Canonical Partition}
For non-elementary graphs,
 the family of maximal barriers never gives a partition of its vertex set~\cite{lp1986}.
Therefore, to analyze the structures of general graphs with perfect matchings, 
we generalized the canonical partition 
based on Kotzig's way~\cite{kotzig1959a, kotzig1959b, kotzig1960}. 

\begin{definition}
Let $G$ be a \matchablesp graph and $H\in \mathcal{G}(G)$.
For $u, v\in V(H)$, we say
$u\gsim v$ if $u = v$ or $G-u-v$ is not \matchable.
\end{definition}
%
\begin{theorem}\label{thm:generalizedcanonicalpartition}
$\gsim $ is an equivalence relation.
\end{theorem}
\begin{proof}
Since the reflexivity and the symmetry are obvious from the definition,
we prove the transitivity.
Let $M$ be a perfect matching of $G$. 
Let $u, v, w\in V(H)$ be such that $u\gsim v$ and  $v\gsim w$.
If any two of them are identical, clearly the claim follows.
Therefore it suffices to consider the case that they are mutually distinct. 
Suppose that the claim fails, that is, $u\not\gsim w$.
Then there is an $M$-saturated path $P$
between $u$ and $w$.
By Proposition~\ref{prop:nonpositive}, there is an $M$-\zero path $Q$ from $v$ to $u$.
Trace $Q$ from $v$ and let $x$ be the first vertex we encounter that in $V(Q)\cap V(P)$.
If $uPx$ has an odd number of edges, 
$vQx + xPu$ is an $M$-saturated path between $u$ and $v$, a contradiction.
If $uPx$ has an even  number of  edges,
then $xPw$ has an odd number of edges, and by the same argument we have a contradiction.
\qed
\end{proof}
We call the family of equivalence classes of $\gsim$
 as the \textit{generalized canonical partition}
and denote as $\pargpart{G}{H}$ 
 for each factor-component $H\in \mathcal{G}(G)$ of a factorizable graph $G$.
Note that
the notions of the  canonical partition and the generalized one are coincident
 for an elementary graph.
 Thus we denote the union of equivalence classes of all the 
factor-components of $G$ as $\gpart{G}$, and call it just as 
the \textit{canonical partition}.
Moreover  our proof for Theorem~\ref{thm:generalizedcanonicalpartition} contains
a short proof for the existence of the canonical partition.
%
Kotzig takes three papers to prove it, thus 
to prove that  $\sim$ is an equivalence relation \textquotedblleft from scratch'' 
is considered to be hard~\cite{lp1986}.
However, in fact, it can be shown in a  simple way
even without the premise of the Gallai-Edmonds structure theorem  nor the notion of barriers. 
Note also that the generalized canonical partition $\pargpart{G}{H}$ is 
a refinement of $\mathcal{P}(H)$ for each $H\in \mathcal{G}(G)$.

\section{Correlations between $\yield$ and  $\gsim$}
In this section we further analyze properties of factorizable graphs.
We denote all the upper bounds of $H\in\mathcal{G}(G)$ 
in $(\mathcal{G}(G), \yield)$ as 
$\parupstar{G}{H}$ and define $\parup{G}{H}$ as $\parupstar{G}{H}\setminus \{H\}$.
We sometimes omit the subscripts 
if they are  apparent from the context.
For simplicity, we sometimes denote the subgraph induced by the vertices in $\up{H}$
 (resp. $\upstar{H}$) as just  
$G[\up{H}]$ (resp. $G[\upstar{H}]$), 
and the vertices of $\up{H}$ (resp. $\upstar{H}$) 
as just $V(\up{H})$ (resp. $V(\upstar{H})$).

\begin{lemma}\label{lem:ear-base}
Let $G$ be a factorizable graph, $M$ be a perfect matching of $G$,
and $H\in \mathcal{G}(G)$.
Let $P$ be an $M$-ear relative to $H$ with end vertices 
$u, v \in V(H)$.
Then $u\gsim v$.
\end{lemma}
\begin{proof}
Suppose the claim fails, that is, $u\neq v$ and 
there is an $M$-saturated path $Q$ between $u$ and $v$.
Trace $Q$ from $u$ and let $x$ be the first vertex we encounter that is  on $P-u$.
If $uPx$ has an even number of  edges,
$uQx + xPu$ is an $M$-alternating circuit containing non-allowed edges,  a contradiction.
Hence we  suppose  $uPx$ has an odd number of  edges.
Let $I\in \mathcal{G}(G)$ be such that  $x\in V(I)$.
Then one of the components of $uQx + xPu - E(I)$ is  an $M$-ear relative to $I$ and through $H$,
a contradiction by Theorem~\ref{thm:equivalentdefinition}.
\qed
\end{proof}  

\begin{theorem}\label{thm:base}
Let $G$ be a factorizable graph, 
and $G_0\in\mathcal{G}(G)$. 
For each connected component $K$ of $G[\up{G_0}]$ 
there exists $T_K\in\pargpart{G}{G_0}$ such that 
$\Gamma(K)\cap V(G_0)\subseteq T_K$. 
\end{theorem}
\begin{proof}
Let $M$ be a perfect matching of $G$. 
\begin{cclaim}\label{claim:sequence2neighbor}
Let $H\in \up{G_0}$,  and  $S$ and $P$ be 
the shortest $M$-ear sequence from $G_0$ to $H$ and 
associated $M$-ears 
which satisfy D1, D2 and D3. 
Then, 
there exists $T\in\pargpart{G}{G_0}$ such that 
for each factor-components $H'$ that 
has common vertices with $\su{S}{P}$,  
$\Gamma(H')\cap V(G_0) \subseteq T$ holds. 
\end{cclaim}
\begin{proof}
Let us denote $S = (G_0 = H_0, \ldots, H_k = H)$, where $k\ge 1$, 
and $P = (P_1,\ldots, P_k)$. 
Let $u_1, v_1 \in V(G_0)$ be the end vertices of $P_1$. 
By Lemma~\ref{lem:ear-base}, 
there exists $T\in\pargpart{G}{G_0}$ such that 
$u_1, v_1\in T$.  

Let $H' \in \mathcal{G}(G)$ be such that $V(H')\cap \su{S}{P} \neq \emptyset$. 
Suppose there exists $w\in \Gamma(H')\cap V(G_0)$ and let 
$z\in V(H')$ be such that $wz\in E(G)$. 
Take $x\in V(H')\cap \su{S}{P}$ arbitrarily. 
By Proposition~\ref{prop:nonpositive}, 
there exists a path $Q$ which is  $M$-balanced from $z$ to $x$ or 
$M$-saturated between $z$ and $x$  such that $V(Q)\subseteq V(H')$. 
Trace $Q$ from $z$ and let $y$ be the first vertex we encounter that 
is in $\su{S}{P}$. 
Then, $zPy$ is an $M$-balanced path from $z$ to $y$ 
with $V(zPy)\subseteq V(H')$ and $V(zPy)\cap \su{S}{P} = \{y\}$. 
By D3($S$, $P$), 
for either of $r\in\{u_1,v_1\}$, 
there is an $M$-balanced path $R$ from $y$ to $r$ 
such that $V(R)\setminus \{r\} \subseteq \su{S}{P}$. 

Therefore, $R + zPy + wz$ forms an $M$-ear relative to $G_0$, 
whose end vertices are $r$ and $w$. 
By Lemma~\ref{lem:ear-base}, therefore, $w\in T$ 
and the claim follows. 
\qed
\end{proof}
Immediately by Claim~\ref{claim:sequence2neighbor} 
we can see that for any $H\in\up{G_0}$ 
there exists $T\in \pargpart{G}{G_0}$ such that 
$\Gamma(H)\cap V(G_0)\subseteq T$. 
Hence for each $T\in\pargpart{G}{G_0}$ we can define 
\begin{quotation}
$\mathcal{K}_T := \{ H\in\up{G_0}: V(H)\subseteq V(K) \mbox{ and } 
\Gamma(H)\cap V(G_0)\subseteq T\}$
\end{quotation}
and $V_T:= \bigcup_{H\in\mathcal{K}_T} V(H)$. 
Note that $\bigcup_{T\in\pargpart{G}{G_0}} V_T = V(K)$. 

We are going to prove the claim by showing that 
$|\{ T\in\pargpart{G}{G_0} : V_T \neq \emptyset\}| = 1$. 
Suppose it fails;  
Then, 
since $K$ is connected, 
there exist $T_1, T_2\in\pargpart{G}{G_0}$ with $T_1\neq T_2$ 
such that $E[V_{T_1}, V_{T_2}] \neq \emptyset$. 
Let $s_1\in V_{T_1}$ and $s_2\in V_{T_2}$ be such that 
$s_1s_2\in E[V_{T_1}, V_{T_2}]$. 
\begin{cclaim}\label{claim:sequence2path}
For each $i = 1, 2$, 
there is an $M$-balanced path $L_i$ from $s_i$ to a vertex in $T_i$, say $r_i$, 
such that $V(L_i)\setminus \{r_i\}\subseteq V_{T_i}$. 
\end{cclaim}
\begin{proof}
Let $i\in \{1,2\}$. 
Let $H\in \mathcal{G}(G)$ be such that $s_i\in V(H)$. 
Then, $V(H)\subseteq V_{T_i}$. 
Take an $M$-ear sequence $S = (G_0 = H_0, \ldots, H_k = H)$, where $k\ge 1$, 
 from $G_0$ to $H$  
 and an associated $M$-ears $P = (P_1,\ldots, P_k)$ which satisfy D1, D2 and D3; 
 By Claim~\ref{claim:sequence2neighbor}, 
$\su{S}{P}\subseteq V_{T_i}$.    
By D3, there is an $M$-balanced path $L_i$ from $s_i$ to 
either of the end vertices of $P_1$, say $r_i\in V(G_0)$ 
such that $V(L_i)\setminus \{r_i\} \subseteq \su{S}{P}$.  
Therefore, 
$V(L_i)\setminus \{r_i\} \subseteq V_{T_i}$. 
\qed
\end{proof}
By Claim~\ref{claim:sequence2path}, 
$L_1 + s_1s_2 + L_2$ is an $M$-ear relative to $G_0$, 
whose end vertices are $r_1\in T_1$ and $r_2\in T_2$. 
By Lemma~\ref{lem:ear-base} 
this yields $T_1 = T_2$, a contradiction. 
Therefore, 
we can conclude that there exists $T\in\pargpart{G}{G_0}$ 
such that $V_T = V(K)$, 
namely the claim follows. 
\qed
\end{proof}

By Theorem~\ref{thm:base},
we can see that upper bounds of a factor-component 
are each ``attached'' to an equivalence class of the generalized canonical partition.
%
\begin{proposition}\label{prop:union}
Let $G$ be a graph and $M$ be a matching of $G$. 
Let $H_1, H_2\subseteq G$ be factor-critical subgraphs of $G$ such that 
there exists $v\in V(H_1)\cap V(H_2)$ and that 
for each $i = 1, 2$, 
$M_{H_i}$ is a near-perfect matching of $H_i$ exposing only $v$. 
Then, $H_1\cup H_2$ is factor-critical. 
\end{proposition}
\begin{proof}
Apparently, $M_1\cup M_2$ is a near-perfect matching of $H_1\cup H_2$, 
exposing only $v$. 
Since $H_1$ and $H_2$ are both factor-critical, 
the claim follows by Proposition~\ref{prop:path2root}. 
\qed
\end{proof}

\begin{lemma}\label{lem:ideal}
Let $G$ be a factorizable graph, and $H\in\mathcal{G}(G)$. 
Then, $G[\upstar{H}]/H$ is factor-critical. 
\end{lemma}
\begin{proof}
Let $M$ be a perfect matching of $G$. 
Let $\mathcal{X}\subseteq 2^{V(G)}$ be the family of 
separable set for $H$. 
Then, by Theorem~\ref{thm:order}, 
$\bigcup_{X\in\mathcal{X}} X = V(\upstar{H})$. 
On the other hand, $G[\bigcup_{X\in\mathcal{X}} X ]/H$ is factor-critical 
by Proposition~\ref{prop:union}. 
Therefore,  the claim follows. 
\qed
\end{proof}

\begin{theorem}
Let $G$ be a factorizable graph, and let $H\in\mathcal{G}(G)$ 
and $S\subseteq \pargpart{G}{H}$. 
Let $K_1,\ldots, K_l$, where $l \ge 1$ be some connected components of $G[\up{H}]$ 
such that $\Gamma(K_i)\cap V(H)\subseteq S$ for $i = 1,\ldots, l$. 
Then, $G[ V(K_1)\cup\cdots\cup V(K_l) \cup S]/S$ is factor-critical. 
\end{theorem}
\begin{proof}
First note that $G[\upstar{H}]/H$ is factor-critical by Lemma~\ref{lem:ideal}.
Let $h$ be the contracted vertex of $G[\upstar{H}]/H$.
Note also that $K$ is a connected component of $G[\up{H}]$ if and only if 
there is a block $\what{K}$ of $G[\upstar{H}]/H$
such that  $K = \what{K} -h$. 
Therefore, by Proposition~\ref{prop:fc_block} 
the claim follows. 
\qed
\end{proof}

\begin{remark}
There are factorizable graphs where 
$\yield$ does not hold for any two factor-components,
in other words, where all the factor-components are minimal in the poset.
For example, we can see by Theorem~\ref{thm:equivalentdefinition} and 
Theorem~\ref{thm:base} that bipartite factorizable graphs are such,
which means Theorem~\ref{thm:order} is not a generalization of 
the DM-decomposition, even though they have similar natures.
\end{remark}

The following theorem shows that 
most of the factorizable graphs with $|\mathcal{G}(G)| \ge 2$, in a sense,
have non-trivial structures as posets. 
%
\begin{theorem}\label{thm:add}
Let $G$ be a factorizable graph,
$G_1, G_2 \in \mathcal{G}(G)$ be
 factor-components for which  $G_1\yield G_2$ does not hold,
  and 
let $G_1$ be minimal in the poset $(\mathcal{G}(G), \yield)$.
Then there are possibly identical complement edges $e, f$ of $G$  between
$G_1$ and $G_2$
such that
$\mathcal{G}(G + e + f)  = \mathcal{G}(G)$ and 
$G_1\yield G_2$ in $(\mathcal{G}(G+e+f), \yield)$.
\end{theorem}

\begin{proof}
First we prove  the case where there is an edge $xy$ such that $x\in V(G_1)$ and $y\in V(G_2)$.
Let $M$ be a perfect matching of $G$.
Choose a vertex $w\in V(G_2)$ such that  $w\not\gsim y$ in $G_2$,
and let $P$ be an $M$-saturated path between $w$ and $y$.
If $xw\in E(G)$,
there is an $M$-ear $xy + P + wx$ relative to $G_1$ and through $G_2$,
which means $G_1\yield G_2$ by Theorem~\ref{thm:equivalentdefinition}.
%
Thus $xw\not\in E(G)$.
%
Suppose $\mathcal{G}(G+xw) \neq \mathcal{G}(G)$.
Then there is an $M$-alternating circuit $C$ containing $xw$ in $G + xw$.
Give an orientation to $C$ so that it becomes a dicircuit with the arc $xx'$.
Trace $C$ from $x$  
and let $z$ be the first vertex we encounter that is in $V(G_2)$. 
Then $xy + xCz$ is an $M$-ear of $G$ which is relative to $G_2$
and through $G_1$, which means $G_2\yield G_1$ by Theorem~\ref{thm:equivalentdefinition},
 a contradiction to the minimality of $G_1$. 
Thus $\mathcal{G}(G+xw) = \mathcal{G}(G)$ and  we are done for this case.

Now we prove the other case,
where there is no edge of $G$ connecting $G_1$ and $G_2$.
Choose any $x\in V(G_1)$ and $y\in V(G_2)$.
If $\mathcal{G}(G + xy ) = \mathcal{G}(G)$, 
we can reduce it to the first case and the claim follows.
%
Therefore it suffices to consider the case that $\mathcal{G}(G + xy ) \neq \mathcal{G}(G)$. 
Then, for any perfect matching $M$ of $G$,
there is an $M$-alternating circuit $C$  in $G+xy$ containing $xy$.
Give an orientation to $C$ so that it becomes a dicircuit with the arc $yy'$.
Trace $C$ from $y$ 
and let $u$ be the first vertex of $G_1$, 
and let $v$ be the first vertex in $G_2$ 
if we trace $C$ from $u$ in the opposite  direction.

If $\mathcal{G}(G + uv) = \mathcal{G}(G)$, the claim follows by the same argument.
Otherwise, that is, if  $\mathcal{G}(G + uv) \neq \mathcal{G}(G)$,
there is an $M$-alternating circuit  $D$ containing $uv$.
Give an orientation to $D$ so that it becomes a dicircuit with the  arc $uu'$.
If $uDv$ is disjoint from the internal vertices of $vCu$, 
then $uDv + vCu$ forms an $M$-alternating circuit containing 
non-allowed edges, a contradiction. 
Otherwise, 
trace $D$ from $u$  
and  let $w$ be the first vertex on $vCu - u$.

If $wCu$ has an even number of edges, 
$wCu + uDw$ is an $M$-alternating circuit of $G$, 
a contradiction.
Therefore, we assume $wCu$ has an odd number of edges.  
Let  $H\in\mathcal{G}(G)$ be such that  $w\in V(H)$.
Then $wCu + uDw - H$ leaves 
an $M$-ear in $G$  which is relative to $H$ and through $G_1$, 
contradicting  
the minimality of $G_1$. Thus this completes the proof. 
\qed
\end{proof}

%
%
\section{Algorithmic Result}
In this section, we discuss the  algorithmic aspects 
of the partial order and the generalized canonical partition.
We denote by $n$ and $m$ respectively the number of vertices 
and edges of input graphs. 
As we work on factorizable graphs and graphs with near-perfect matchings,
we can assume $m = \mathrm{\Omega}(n)$.

We start with  some materials from Edmonds' maximum matching algorithm~\cite{edmonds1965},
referring mainly to~\cite{lp1986, kv2008}.
For a tree $T$ with a specified root vertex $r$, 
we call a vertex  $v\in V(T)$ \textit{inner} (resp. \textit{outer})
if the unique path  in $T$ from $r$ to $v$ has  an odd (resp. even) number of edges.
Let $G$ be a graph and $M$ be a matching of $G$.
A tree $T\subseteq G$  is called \textit{$M$-alternating} if 
exactly one vertex of it, the  root, is exposed by $M$ in $G$, and
each inner vertex $v\in V(T)$ satisfies  
$|\delta(v)\cap E(T)|= 2$  and one of the edges of $\delta(v)\cap E(T)$ 
is contained in $M$.

A subgraph $S \subseteq G$ is called a \textit{special blossom tree with respect to $M$}
(\textit{$M$-SBT})
if  there is a partition  
 $V(C_1)\dot{\cup}\cdots\dot{\cup} V(C_k) = V(S)$ such that 
\begin{enumerate}
\item $S' := S/C_1/\cdots /C_k$ is an $M$-alternating tree,
\item $M_{C_i}$ is a near-perfect matching of $C_i$,
\item $C_i$ is a maximal factor-critical subgraph of $G$ if 
it corresponds to an outer vertex of $S'$, 
and called  an \textit{outer blossom}, and 
\item $|V(C_i)|>1$ only if $C_i$ is an outer blossom,
for each $i = 1,\ldots ,k$.
\end{enumerate}
%
%
%

Edmonds' maximum matching algorithm  tells us the following facts.
Let $G$ be a graph, $M$ be a near-perfect matching of $G$,
and $r\in V(G)$ be the vertex exposed by $M$.
Then an $M$-SBT $S$,  with root $r$,
can be computed, if it is carefully implemented~\cite{tarjan1983, gt1985},
in $O(m)$ time.
Additionally, the set of vertices from which $r$ can be reached 
by an $M$-balanced path is exactly 
the set of vertices contained in the outer blossoms of $S$.

%
%
Thus, due to an easy reduction of the above facts,
the following proposition holds; they can be regarded as a 
folklore. See \cite{cc2005}.
(In \cite{cc2005} they  are presented as 
those for elementary graphs, but in fact,
they can be applicable for general factorizable graphs.) 
\begin{proposition}\label{prop:pathalg}
Let $G$ be a factorizable graph, $M$ be a perfect matching of $G$,
and $u\in V(G)$.
\begin{enumerate}
\item  The set of vertices that can be reached from $u$ by 
an $M$-saturated path can be computed in $O(m)$ time.
\item All the allowed edges adjacent to $u$ can be computed in $O(m)$ time.
\item All the factor-components of $G$ can be computed in $O(nm)$ time.
\end{enumerate}
\end{proposition}
%
%
%

\begin{proposition}\label{prop:partalg}
Given  a factorizable graph $G$, one of its perfect matchings $M$ and $\mathcal{G}(G)$, 
we can compute the generalized canonical partition of $G$ 
in $O(nm)$ time.
\end{proposition}

\begin{proof}
For each  $H\in \mathcal{G}(G)$,
we can compute $\pargpart{G}{H}$ 
 in a similar way to compute the canonical partition of 
an elementary graph~\cite{cc2005}.
That is, 
for each $v\in V(H)$,
compute the set of vertices $U$ that can be reached from $v$ by an 
$M$-saturated path,
and recognize $V(H)\setminus U$ as a member of $\pargpart{G}{H}$.
This procedure is surely compatible 
by  Theorem~\ref{thm:generalizedcanonicalpartition}.
Thus, the claim follows by Proposition~\ref{prop:pathalg}.
\qed
\end{proof}

Let $G$ be a factorizable graph and $M$ be a perfect matching of $G$.
We say two distinct factor-components $G_1, G_2$ of $G$ with $G_1\yield G_2$
 are \textit{non-refinable}
if $G_1\yield H \yield G_2$ yields $G_1 = H$ or $G_2 = H$ for any $H\in \mathcal{G}(G)$.
Note that if $G_1$ and $G_2$ are non-refinable,
then there is an $M$-ear relative to $G_1$ and through $G_2$ 
by Theorem~\ref{thm:equivalentdefinition}.
Note also that the converse of the above fact does not hold.
\begin{lemma}\label{lem:orderalg}
Let $G$ be a factorizable graph, $M$ be a perfect matching of $G$,
 and $H\in \mathcal{G}(G)$.
Let $S$ be a maximal $M$-SBT in $G/H$ 
 and let
$C$ be the blossom of $T$  containing the contracted vertex $h$ corresponding to $H$.
Then any non-refinable upper bound  of $H$ in $(\mathcal{G}(G), \yield)$ 
has common vertices with $C$.
Additionally, if a factor-component $I\in\mathcal{G}(G)$ 
has some common vertices with $C$,
then $H\yield I$.   
\end{lemma}

\begin{proof}
For the former part, 
let $H'$ be a non-refinable upper bound of $H$, and $P$ be an $M$-ear relative to $H$
and through $H'$.
Since $P-C$ is a disjoint union of $M$-ears relative to $C$,
we have $P\subseteq C$ 
by Theorem~\ref{thm:odd} and the maximality of the outer blossoms in $M$-SBT.
Thus the former part of the claim follows.

For the latter part,
by the definition of $M$-SBT and Proposition~\ref{prop:fc_alt}, 
there is an $M$-alternating odd ear-decomposition
$\mathcal{P} = \{ P_1, \ldots , P_k\}$ of $C$.
Let $I\in\mathcal{G}(G)$ be such that  $V(I)\cap V(C)\neq \emptyset$
and that 
$V(P_j)\cap V(I) = \emptyset$ for $j = 1,\ldots, i-1$
and $V(P_i)\cap V(I) \neq  \emptyset$.
%
%
We proceed by induction on $i$.
If $i = 1$, the claim  obviously follows.
Let  $i > 1$.
$G_{i-1} := P_1+\cdots + P_{i-1}$ is factor-critical by Theorem~\ref{thm:odd},
and $M_{G_{i-1}}$ is a near-perfect matching of $G_{i-1}$.
Moreover, $P_i$ is an $M$-ear relative to $G_{i-1}$.
Therefore, with the same technique as in the proof of Theorem~\ref{thm:equivalentdefinition},
there exists $I'\in\mathcal{G}(G)$ 
such that  $V(I')\cap V(C) \neq \emptyset $ and 
 that there is an $M$-ear relative to $I'$ and through $I$.
 Thus, by the induction hypothesis,  the latter part of the  claim follows.
 \qed
\end{proof}
\begin{proposition}\label{prop:orderalg}
Given a factorizable graph $G$, its perfect matching $M$, and $\mathcal{G}(G)$,
we can compute the poset $(\mathcal{G}(G), \yield)$ in $O(nm)$ time.
\end{proposition}
\begin{proof}
It is sufficient 
to list all the non-refinable upper bounds for each factor-component of $G$
by the following procedure:

\begin{algorithmic}[1]
\STATE $D := (\mathcal{G}(G),\emptyset)$; $A:=\emptyset$;
\FORALL{$H\in\mathcal{G}(G)$}
\STATE compute a maximal $M$-SBT $T$; let $C$ be the blossom of $T$ corresponding to its root;
\FORALL{$x\in V(C)$, which satisfies $x\in V(I)$ for some  $I\in\mathcal{G}(G)$}
\STATE $A:= A\cup \{(H, I)\}$;
\ENDFOR
\ENDFOR
\STATE $D := (\mathcal{G}(G), A)$; STOP.
\end{algorithmic}
By Lemma~\ref{lem:orderalg}, 
the partial order on $V(D)$ determined by the reachability
 corresponds to $\yield$ after the above procedure.
 That is, if we define a binary relation $\prec$  on $V(D)$ so that 
 $H'\yield I'$ if there is a dipath from $H'$ to $I'$ in $D$,
 then $\prec$ and $\yield$ coincide.
For each $H\in \mathcal{G}(G)$, the above procedure costs $O(m)$ time,
thus it costs $O(nm)$ time over the whole computation.
\qed
\end{proof}
\begin{remark}
Given the digraph $D$ after the procedure in Proposition~\ref{prop:orderalg},
we can compute all the upper bounds of a factor-component in 
$O(n^2)$ time.
Thus, an efficient data structure that represents the poset, for example,
a boolean-valued matrix $L$ where $L[i, j] = \mathrm{true}$ 
if and only if $G_i\yield G_j$,
can be obtained in additional $O(n^2)$ time.
\end{remark}  
As a maximum matching of a graph can be computed 
in $O(\sqrt{n}m)$ time~\cite{mv1980, vazirani1994},
we have the following, combining Propositions~\ref{prop:pathalg},~%
\ref{prop:partalg}, and \ref{prop:orderalg}.
\begin{theorem}
Let $G$ be a factorizable graph.
Then the poset $(\mathcal{G}(G), \yield)$ and 
the generalized canonical partition $\gpart{G}$ can be computed in 
$O(nm)$ time.
\end{theorem}

%
%
%
%
\begin{ac}
 The author is grateful to Yusuke Kobayashi and Richard Hoshino for 
carefully reading the paper, and Akihisa Tamura for useful discussions. 
\end{ac}
\bibliographystyle{splncs}
\bibliography{isaac2012_full.bib}

%
\end{document}